\documentclass[conference]{IEEEtran}
\IEEEoverridecommandlockouts
\usepackage{mathtools}	
\usepackage{amssymb}
\usepackage{amsthm}

\usepackage[utf8]{inputenc}

\usepackage{ifthen}

\usepackage{microtype}

\usepackage{caption}

\usepackage{bm}

\usepackage[draft,nomargin,marginclue,inline]{fixme}
\makeatletter
\renewcommand*\FXLayoutMarginClue[3]{%
	\marginpar[%
	\raggedleft\@fxuseface{margin}\textcolor{red}{\ignorespaces $ \Rightarrow $}]{%
		\raggedright\@fxuseface{margin}\textcolor{red}{\ignorespaces $ \Leftarrow $}}}
\makeatother	

\usepackage{tumcolor}

\usepackage{pgfplotstable}	

\pgfplotsset{
	discard if/.style 2 args={
		x filter/.append code={
			\edef\tempa{\thisrow{#1}}
			\edef\tempb{#2}
			\ifx\tempa\tempb
			
			\fi
		}
	},
	discard if not/.style 2 args={
		x filter/.append code={
			\edef\tempa{\thisrow{#1}}
			\edef\tempb{#2}
			\ifx\tempa\tempb
			\else
			
			\fi
		}
	}
}

\usepackage{cite}
\usepackage{hyperref}
\hypersetup{
	colorlinks = true,
	linkbordercolor = {white},
	citecolor = {black},
	linkcolor = {black},
	urlcolor = {black}
}
\makeatletter
\renewcommand*{\eqref}[1]{%
	\hyperref[{#1}]{\textup{\tagform@{\ref*{#1}}}}%
}
\makeatother

\usepackage[acronym,shortcuts]{glossaries}
\newacronym{blmmse}{BLMMSE}{Bussgang LMMSE}
\newacronym{cnn}{CNN}{convolutional neural network}
\newacronym{dft}{DFT}{discrete Fourier transform}
\newacronym{em}{EM}{expectation-maximization}
\newacronym{emiht}{EM-IHT}{EM algorithm with IHT}
\newacronym{iht}{IHT}{iterative hard thresholding}
\newacronym{mimo}{MIMO}{multiple-input multiple-output}
\newacronym{mmse}{MMSE}{minimum mean square error}
\newacronym{mse}{MSE}{mean square error}
\newacronym{relu}{ReLU}{rectified linear unit}
\newacronym{snr}{SNR}{signal-to-noise ratio}
\newacronym{adc}{ADC}{analog-to-digital converter}
\newacronym{nn}{NN}{neural network}
\newacronym{re}{RE}{resource element}

\newcommand{\op}[1]{{\operatorname{#1}}}
\newcommand{\uproman}[1]{\uppercase\expandafter{\romannumeral#1}}

\newcommand{\diag}{\operatorname{diag}}

\newcommand{\eye}{\bm{\op{I}}}
\newcommand{\B}[1]{\bm{#1}}
\newcommand{\inv}{^{-1}}

\newcommand{\h}{^{\op H}}

\usepackage{cleveref}

\usepackage{enumitem}

\tikzset{algorithm1/.style={mark options={solid},color=TUMBeamerBlue, line width=\lineWidth, mark=square, dashed}}

\newcommand*{\C}{\mathbb{C}}

\newlength{\leftstackrelawd}
\newlength{\leftstackrelbwd}
\def\leftstackrel#1#2{\settowidth{\leftstackrelawd}%
	{${{}^{#1}}$}\settowidth{\leftstackrelbwd}{$#2$}%
	\addtolength{\leftstackrelawd}{-\leftstackrelbwd}%
	\leavevmode\ifthenelse{\lengthtest{\leftstackrelawd>0pt}}%
	{\kern-.5\leftstackrelawd}{}\mathrel{\mathop{#2}\limits^{#1}}}

\newcommand{\mbn}{\bm{n}}

\newcommand{\mbs}{\bm{s}}

\newcommand{\mby}{\bm{y}}

\usepackage{algorithmic}
\usepackage{graphicx}
\pgfplotsset{compat=1.15}
\usepgfplotslibrary{statistics}
\usetikzlibrary{patterns}
\usetikzlibrary{positioning}

\usepackage{siunitx}
\usepackage{tabularx}
\usepackage{subcaption}
\usepackage{balance}

\newacronym{gmm}{GMM}{Gaussian mixture model}
\newacronym{pdf}{PDF}{probability density function}
\newacronym{cme}{CME}{conditional mean estimator}
\newacronym{siso}{SISO}{single-input single-output}
\newacronym{fft}{FFT}{fast Fourier transform}
\newacronym{ofdm}{OFDM}{orthogonal frequency-division multiplexing}
\newacronym{ml}{ML}{machine learning}
\newacronym{mt}{MT}{mobile terminal}
\newacronym{bs}{BS}{base station}
\newacronym{ul}{UL}{uplink}
\newacronym{dl}{DL}{downlink}
\newacronym{fdd}{FDD}{frequency division duplex}
\newacronym{tdd}{TDD}{time division duplex}
\newacronym{ce}{CE}{channel estimation}
\newacronym{los}{LOS}{line of sight}
\newacronym{nlos}{NLOS}{non-line of sight}
\newacronym{o2i}{O2I}{outdoor-to-indoor}
\newacronym{uma}{UMa}{urban macrocell}
\newacronym{umi}{UMi}{urban microcell}
\newacronym{mfa}{MFA}{mixtures of factor analyzers}
\newacronym{fa}{FA}{factor analysis}
\newacronym{ura}{URA}{uniform rectangular array}
\newacronym{lte}{LTE}{Long Term Evolution}
\newacronym{awgn}{AWGN}{additive white Gaussian noise}
\newacronym{vae}{VAE}{variational autoencoder}
\newacronym{pca}{PCA}{principle component analysis}
\newacronym{hmm}{HMM}{hidden Markov model}
\newacronym{gan}{GAN}{generative adversarial network}
\newacronym{csi}{CSI}{channel state information}
\newacronym{simo}{SIMO}{single-input multiple-output}
\newacronym{ls}{LS}{least squares}
\newacronym{omp}{OMP}{orthogonal matching pursuit}
\newacronym{lmmse}{LMMSE}{linear \ac{mmse}}
\newacronym{ris}{RIS}{reconfigurable intelligent surface}

\def\BibTeX{{\rm B\kern-.05em{\sc i\kern-.025em b}\kern-.08em
    T\kern-.1667em\lower.7ex\hbox{E}\kern-.125emX}}

\newcommand{\plotwidth}{1\columnwidth}

\newcommand{\lineWidth}{1.2pt}
\newcommand{\marksize}{1.6pt}
\pgfplotsset{tick label style={font=\small},label style={font=\small},legend style={font=\scriptsize}}

\definecolor{myblack}{RGB}{70,70,70}
\definecolor{myblue}{RGB}{65,105,225}
\definecolor{mygreen}{RGB}{0,139,139}
\definecolor{myorange}{RGB}{255,150,0}
\definecolor{myred}{RGB}{255,69,0}
\definecolor{mylila}{RGB}{153,50,204}

\tikzset{ls/.style={mark options={solid},color=black, line width=\lineWidth}}
\tikzset{sampcov/.style={mark options={solid},color=gray, line width=\lineWidth, mark=triangle, mark size=\marksize}}
\tikzset{omp/.style={mark options={solid},color=orange, line width=\lineWidth, mark=pentagon, mark size=\marksize}}
\tikzset{cnn/.style={mark options={solid},color=TUMBeamerRed, line width=\lineWidth, mark=square, mark size=\marksize,dotted}}
\tikzset{gmm/.style={mark options={solid},color=TUMBeamerGreen, line width=\lineWidth, mark=diamond, mark size=\marksize}}
\tikzset{gmm_toep/.style={mark options={solid},color=TUMBeamerGreen, line width=\lineWidth, mark=diamond, mark size=\marksize,dashed}}
\tikzset{gmm_circ/.style={mark options={solid},color=TUMBeamerGreen, line width=\lineWidth, mark=diamond, mark size=\marksize,dotted}}
\tikzset{mfa/.style={mark options={solid},color=TUMBlue, line width=\lineWidth, mark=o, mark size=\marksize}}
\tikzset{mfa2/.style={mark options={solid},color=TUMBlue, line width=\lineWidth, mark=o, mark size=\marksize,dashed}}
\tikzset{mfa3/.style={mark options={solid},color=TUMBlue, line width=\lineWidth, mark=o, mark size=\marksize,dotted}}

\newtheorem{corollary}{Corollary}

\begin{document}
\bstctlcite{IEEEexample:BSTcontrol}

\title{Low-Rank Structured MMSE Channel Estimation with Mixtures of Factor Analyzers
\thanks{The authors acknowledge the financial support by the Federal Ministry of
Education and Research of Germany in the program of ``Souver\"an. Digital.
Vernetzt.''. Joint project 6G-life, project identification number: 16KISK002.}
}
\author{
	\centerline{Benedikt Fesl, Nurettin Turan, and Wolfgang Utschick}\\
	\IEEEauthorblockA{School of Computation, Information and Technology, Technical University of Munich, Germany\\
	Email: \{benedikt.fesl, nurettin.turan, utschick\}@tum.de
    }
}

\maketitle

\begin{abstract}
    This work proposes a generative modeling-aided channel estimator based on \ac{mfa}.
    In an offline step, the parameters of the generative model are inferred via an \ac{em} algorithm in order to learn the underlying channel distribution of a whole communication scenario inside a \ac{bs} cell.
    Thereby, the wireless channels are effectively modeled on a piecewise linear subspace which is achieved by the low-rank structure of the learned covariances of the \ac{mfa}. This suits the low-rank structure of wireless channels at high frequencies and additionally saves parameters and prevents overfitting.
    Afterwards, the trained \ac{mfa} model is used online to perform channel estimation with a closed-form solution of the estimator which asymptotically converges to the \ac{mmse} estimator.
    Numerical results based on real-world measurements demonstrate the great potential of the proposed approach for channel estimation.
\end{abstract}


\section{Introduction}
The concept of generative models has been existing for a long time.~\cite{Roweis1999}. 
One prominent example is the \ac{gmm}, which has the ability for universal approximation \cite{NgNgChMc20}.
With the advent of deep learning, generative models based on neural networks, such as \acp{vae} \cite{Kingma2014}, have become highly successful.
Therefore, generative models play a key role in modern signal processing and wireless communication applications, particularly when domain and system knowledge is incorporated to solve inference tasks \cite{10056957}.
In this respect, the requirements for accurate channel estimation in the next generation of cellular systems (6G) are of significant interest \cite{9598915}.

In recent works, both \acp{gmm} and \acp{vae} were leveraged to learn the underlying channel distribution of a whole communication scenario and to utilize this information to yield a tractable prior information for channel estimation, showing great improvements over state-of-the-art approaches \cite{9842343,10051858}. The advantages are the possibilities to incorporate structural features and even learn from imperfect data \cite{10051921,10078293}.
A key feature of both approaches is the parameterization of the local distribution of the \acp{mt} inside a \ac{bs} cell in a tractable manner by the learned parameters. This allows to utilize closed-form solutions for the estimation task. 

Although these methods show promising results, they face some challenges which can potentially hinder the application in real systems.
First, the number of learned parameters is high, entailing demanding memory requirements.
For instance, the number of parameters scales quadratically in the number of antennas for the \ac{gmm}; Likewise, the \ac{vae} is comprised of deep \acp{nn} with typically even more parameters. 
While structural features imposed by antenna arrays can reduce the number of parameters, mutual coupling between the antennas or other hardware imperfections can corrupt this structure in real systems.
Second, the generative abilities of \acp{gmm} are known to be somewhat limited due to the discrete nature of the latent space, whereas \acp{vae} are generally lacking interpretability due to the elaborate design of nested nonlinearities. 

A powerful generative model that is related to \acp{gmm} and \acp{vae} is the \ac{mfa} model, which contains both a discrete and continuous latent variable, where the latter is modeled Gaussian just as in the \ac{vae} \cite[Ch. 12]{Murphy2012}, \cite{em_mfa}. 
This motivated the usage of \ac{mfa} for several applications, e.g. \cite{deepMFA,8645457}.
In contrast to the \ac{vae} with a nonlinear latent space, the latent space of the \ac{mfa} is piecewise linear with tractable expressions for the inference of the latent samples. From a different perspective, the \ac{mfa} model can be interpreted as a \ac{gmm} with low-rank structured covariances \cite[Ch. 12]{Murphy2012}. This results in having less parameters and being less prone to overfitting, thereby matching the structural features of channels at high frequencies \cite{7727995}. 
Altogether, the \ac{mfa} model has the potential to excellently match the requirements for accurate channel estimation in 6G systems with low~overhead.

\textit{Contributions:}
We propose to employ the \ac{mfa} model to learn the unknown underlying channel distribution of a whole \ac{bs} cell with low-rank structured covariances which effectively models the channel distribution on a piecewise linear subspace. 
The resulting channel estimator can be computed in closed form by a convex combination of \ac{lmmse} estimates, which asymptotically converges to the generally intractable \ac{mse}-optimal solution. 
We validate the effectiveness of the approach through simulation results based on real-world measurement data, which demonstrate that the MFA achieve great results in channel estimation performance.

\section{System Model and Measurement Campaign}\label{sec:system_model}
We consider a \ac{simo} communication scenario where the \ac{bs} equipped with $N$ antennas receives uplink training signals from a single-antenna \ac{mt}. In particular, at
the \ac{bs}, after decorrelating the pilot signal, noisy observations of the form 
\begin{equation}
    \B y = \B h + \B n \in \mathbb{C}^N
    \label{eq:system_model}
\end{equation}
are received where the channel $\B h \in \mathbb{C}^N$, which follows an unknown distribution $p(\B h)$, is corrupted by \ac{awgn} $\B n \sim \mathcal{N}_{\mathbb{C}}(\B 0 , \sigma^2\eye)$.
Although the underlying channel distribution $p(\B h)$ is unknown, we assume the availability of a training dataset $\mathcal{H} = \{\B h_t\}_{t=1}^T$ of $T$ channel samples that represent the channel distribution of the whole \ac{bs} cell. A common practice is to use simulation tools which are based on sophisticated models of the underlying communication scenario to generate a dataset, however, these models do not fully capture the characteristics of the real world.
Therefore, we use real-world data from a measurement campaign which is described in the following.

The measurement campaign was conducted at the Nokia campus in Stuttgart, Germany, in October/November 2017.
As can be seen in Fig. \ref{fig:meas_campaign}, the receive antenna with a down-tilt of $\SI{10}{\degree}$ was mounted on a rooftop about $\SI{20}{m}$ above the ground and comprises a \ac{ura} with $N_v=4$ vertical and $N_h=16$ horizontal single polarized patch antennas.
The horizontal spacing is $\lambda/2$ and the vertical spacing equals $\lambda$, where the geometry of the \ac{bs} antenna array was adapted to the \ac{umi} propagation scenario.
The carrier frequency is $\SI{2.18}{\giga\hertz}$.
The \ac{bs} transmitted time-frequency orthogonal pilots using $\SI{10}{\mega\hertz}$ \ac{ofdm} waveforms.
In particular, $600$ sub-carriers with $\SI{15}{\kilo\hertz}$ spacing were used, which resembles typical \ac{lte} numerology.
The pilots were sent continuously with a periodicity of $\SI{0.5}{ms}$ and were arranged in $50$ separate subbands, with $12$ consecutive subcarriers each, for channel sounding purposes.
For the duration of one pilot burst the propagation channel was assumed to remain constant.
A single monopole receive antenna, which mimics the \ac{mt}, was mounted on top of a moving vehicle at a height of $\SI{1.5}{m}$.
The maximum speed was $\SI{25}{kmph}$.
Synchronization between the transmitter and receiver was achieved using GPS.
The data was collected by a TSMW receiver and stored on a Rohde \& Schwarz IQR hard disk recorder.
In a post-processing step, by the correlation of the received signal with the pilot sequence a channel realization vector with $N=N_v N_h$ coefficients per subband was extracted.

The measurement was conducted at a high \ac{snr}, which ranged from $\SI{20}{dB}$ to $\SI{30}{dB}$.
Thus, the measured channels are regarded as ground truth.
In this work, we will therefore consider a system where we artificially corrupt the measured channels with \ac{awgn} at specific \acp{snr} and thereby obtain noisy observations $\B y = \B h + \B n$.
We note that we investigate a single-snapshot scenario, i.e., the
coherence interval of the covariance matrix and of the channel is identical. 

\begin{figure}[t]
    \centering
    \begin{tikzpicture}
        \draw (0,0) node[below right] {\includegraphics[width=0.9\columnwidth]{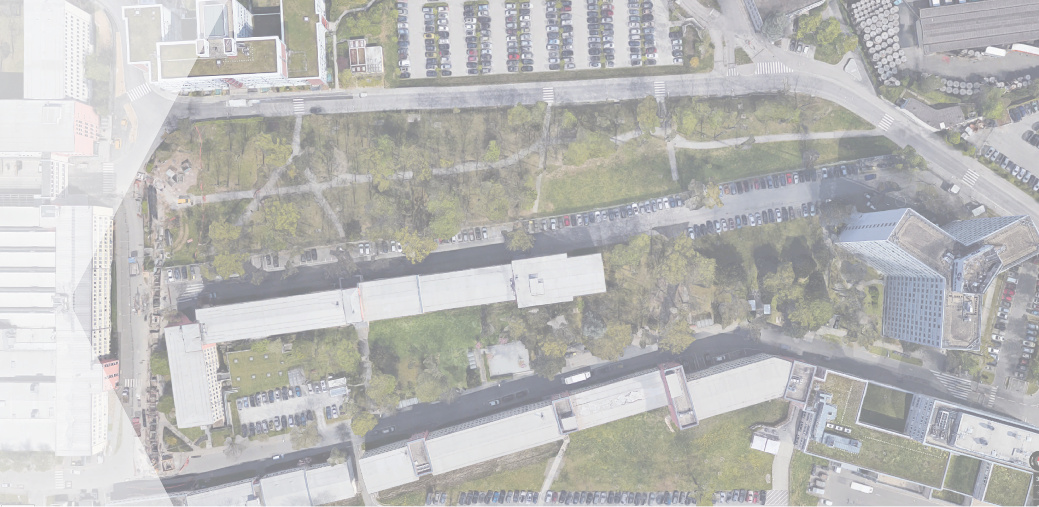}};
        \draw (0,0) node[below right] {\includegraphics[width=0.9\columnwidth]{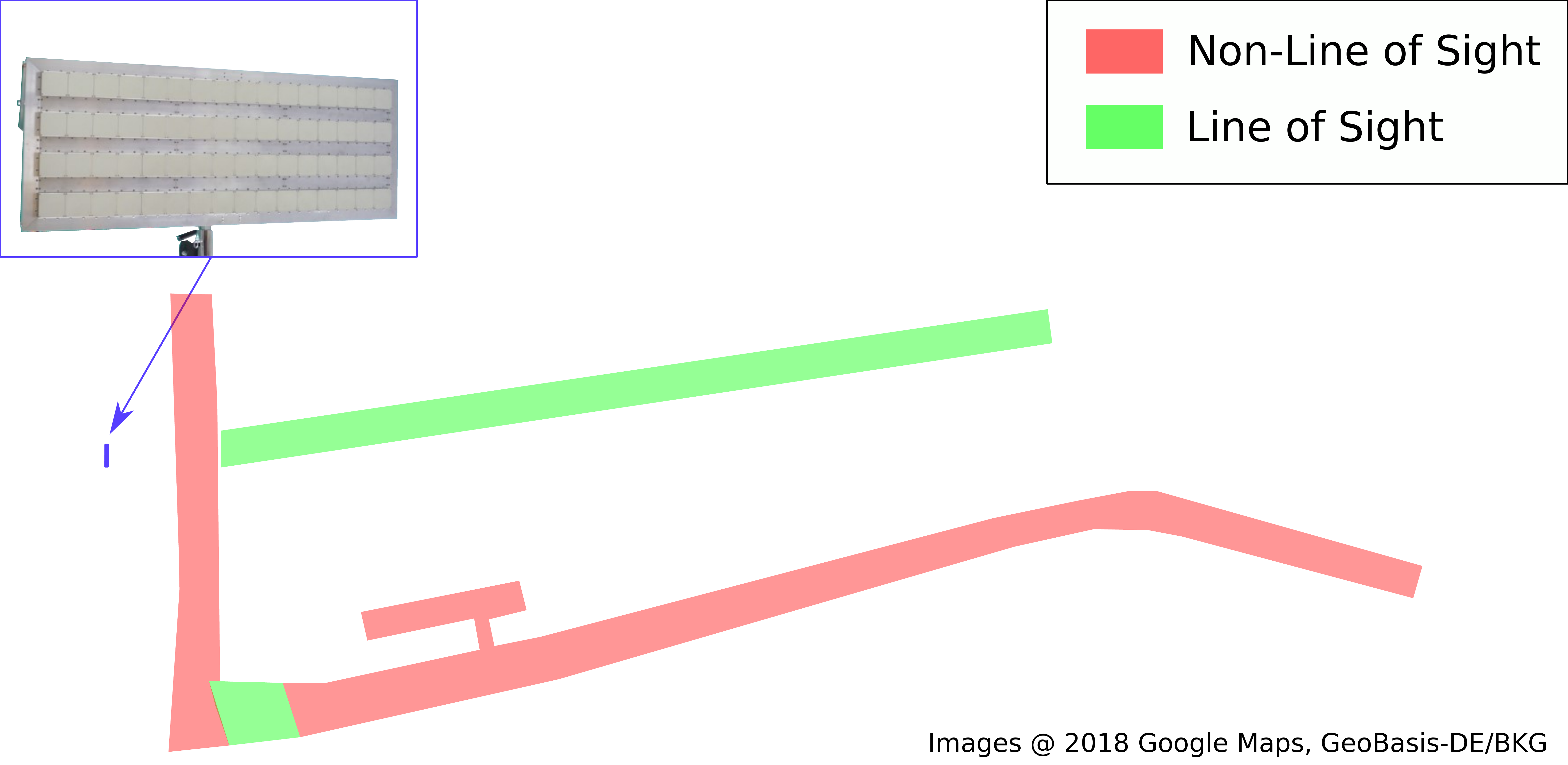}};
    \end{tikzpicture}
    \caption{Measurement setup on the Nokia campus in Stuttgart, Germany.}
    \label{fig:meas_campaign}
\end{figure}

\section{Mixture of Factor Analyzers}\label{sec:mfa}
We start by briefly revising the \ac{fa} which is the basis for the \ac{mfa}. 
The generative model is given as
\begin{equation}
    \B h^{(L)} = \B W \B z + \B u
    \label{eq:generator_fa}
\end{equation}
where $\B W\in \mathbb{C}^{N\times L}$ is the so called factor loading matrix, $\B z \in \mathbb{C}^{L}$ is the latent variable, and $\B u\sim \mathcal{N}_{\mathbb{C}}(\B 0, \B \Psi)$ is an additive term. 
The prior distribution is assumed to be standard Gaussian, i.e., $p(\B z)= \mathcal{N}_\mathbb{C}(\B z; \B 0, \eye)$. The key assumptions for this model are that the latent variable $\B z$ is low-dimensional, i.e., $L< N$ holds, and that the covariance $\B \Psi$ is diagonal. 
The rationale is that the data are modeled on a linear subspace described by $\B W$, which explains the common factors/features and their correlations, and the additional term $\B u$ accounts for both unique factors and noise in the data \cite[Ch. 12]{Murphy2012}. Note that if $\B \Psi = \psi^2\eye$ or $\B \Psi = \B 0$, the \ac{fa} degenerates to the (probabilistic) \ac{pca}. Moreover, the \ac{fa} is a low-rank parameterization of a Gaussian, i.e., $\B h^{(L)} \sim \mathcal{N}_{\mathbb{C}}(\B 0, \B W\B W\h + \B \Psi)$, cf. \cite[Ch. 12]{Murphy2012}.

Since the linearity of the latent space and the Gaussian assumption of the data is too restrictive for modeling real-world communication scenarios, we aim for a more powerful generative model which is realized by the \ac{mfa}. Thereby, a mixture of $K$ linear subspaces is considered which yields the following generative model for the $k$th mixture component:
\begin{align}
    \B h^{(K,L)} \mid k = \B W_k \B z + \B u_k + \B \mu_k
    \label{eq:generator_mfa}
\end{align}
where $\B W_k \in \mathbb{C}^{N\times L}$ is the factor loading matrix of mixture $k$, $\B u_k\sim \mathcal{N}_{\mathbb{C}}(\B 0, \B \Psi_k)$, and $\B \mu_k$ is the mean of mixture $k$.
The resulting generative model then follows the distribution
\begin{align}
    p^{(K,L)}(\B h) = \sum_{k=1}^K \int p(\B h \mid \B z,k) p(\B z\mid k) p(k) \op d \B z
    \label{eq:full_dist}
\end{align}
where $ p(\B z |k) = p(\B z) = \mathcal{N}_{\mathbb{C}}(\B 0, \eye)$
and $p(k)$ follows a categorical distribution \cite[Ch. 12]{Murphy2012}.
An important property for our considerations is that, given the latent variables, the distribution is conditionally Gaussian, i.e.,
\begin{align}
    p^{(K,L)}(\B h \mid \B z, k) &= \mathcal{N}_{\mathbb{C}}(\B h; \B \mu_k + \B W_k \B z, \B \Psi_k).
\end{align}
By integrating out the latent variable $\B z$ in \eqref{eq:full_dist}, one can interpret the \ac{mfa} as a \ac{gmm} with low-rank structured covariances, i.e.,
\begin{align}
    p^{(K,L)}(\B h) = \sum_{k=1}^Kp(k)\mathcal{N}_{\mathbb{C}}(\B h; \B \mu_k, \B W_k\B W_k\h + \B \Psi_k).
    \label{eq:gmm}
\end{align}
This model has generally much less parameters than a \ac{gmm} with full covariances since $L<N$.
Note that there exist different restrictions on the diagonal covariance matrix $\B \Psi_k$, e.g., it is possible to choose a common matrix for all components as $\B \Psi_k = \B \Psi,~\forall k\in\{1,\dots,K\}$, or to choose a scaled identity $\B \Psi_k = \psi_k^2\eye$, yielding different degrees of freedom.

Given a training dataset $\mathcal{H}$, an \ac{em} algorithm can be used to fit the parameters of the \ac{mfa} model \cite{em_mfa}. Interestingly, the \ac{mfa} model meets all formal requirements for the universal approximation property from \cite{NgNgChMc20} in the sense that, for an infinite number of mixture components, any continuous \ac{pdf} can be asymptotically approximated arbitrarily well by means of \ac{mfa} as 
\begin{align}
    \lim_{K\to \infty} \|p(\B h) - p^{(K,L)}(\B h)\|_\infty = 0.
    \label{eq:univ_approx}
\end{align}
It is worth noting that this result holds for an any number $L$ of latent dimensions which is a powerful basis for our consideration of learning a tractable approximation of the underlying channel distribution by means of the \ac{mfa}. In view of this property, in this work we investigate the asymptotic behavior of the \ac{mfa}-based channel estimator and the relationship between the number of latent dimensions $L$ and the number of mixture components $K$ for a given number of training samples $T$.

\section{Channel Estimation}\label{sec:ch_est}
In this section, we derive an approximation of the generally intractable \ac{mmse} estimator via the \ac{mfa} model.
The \ac{mse}-optimal estimator for an arbitrary channel distribution is the \ac{cme} which is given as
\begin{align}
    \hat{\B h}_{\text{CME}}(\B y) =  \op E [\B h \mid \B y] = \int \B h p(\B h \mid \B y) \op d \B h.
    \label{eq:cme}
\end{align}
Note that this estimator cannot be computed generally since the conditional distribution $p(\B h\mid \B y)$ is unknown. Even if the channel distribution $p(\B h)$ would be known, computing the \ac{cme} via the integral is intractable in real-time systems. To this end, one needs tractable expressions of the involved distributions such that the resulting estimator can be computed with a manageable complexity.

\subsection{MFA-based Channel Estimator}\label{subsec:ch_est}
We aim to find a tractable expression of the \ac{cme} by using the learned \ac{mfa} model and introduce the discrete latent variable via the law of total expectation:
\begin{align}
    \hat{\B h}^{(K,L)}(\B y) = \op E [\B h^{(K,L)} \mid \B y] &= \op E \left[\op E [\B h^{(K,L)} \mid \B y,k] \mid \B y \right].
    \label{eq:total_exp1}
\end{align}
We note that, since conditioned on the component $k$ the model is Gaussian, i.e., $\B h^{(K,L)}\mid k \sim \mathcal{N}_{\mathbb{C}}(\B \mu_k, \B W_k\B W_k\h + \B \Psi_k)$, cf. \eqref{eq:gmm}, we can use the well-known \ac{lmmse} formula to solve the inner expectation in closed form, which yields
\begin{equation}
    \begin{aligned}
        \op E [\B h^{(K,L)} \mid \B y,k] &= \B \mu_k + (\B W_k\B W_k\h + \B \Psi_k)
        \\
        &\left(\B W_k\B W_k\h + \B \Psi_k + \sigma^2\eye\right)\inv(\B y - \B \mu_k).
    \end{aligned}
    \label{eq:lmmse}
\end{equation}

The outer expectation in \eqref{eq:total_exp1} is then, due to the discrete nature of the latent variable, a convex combination of the linear filters from \eqref{eq:lmmse}, given as
\begin{align}
    \hat{\B h}^{(K,L)}(\B y) = \sum_{k=1}^K  p(k \mid \B y) \op E[\B h^{(K,L)} \mid \B y,k]
    \label{eq:mfa_estimator}
\end{align}
where $p(k\mid \B y)$ is the responsibility of the $k$th component for the pilot observation $\B y$ which is computed as
\begin{align}
    p(k \mid \B y) = \frac{p(k)\mathcal{N}_{\mathbb{C}}(\B y; \B \mu_k, \B W_k\B W_k\h + \B \Psi_k + \sigma^2\eye) }{\sum_{i=1}^K p(i)\mathcal{N}_{\mathbb{C}}(\B y; \B \mu_i, \B W_i\B W_i\h + \B \Psi_i + \sigma^2\eye)}.
    \label{eq:resp}
\end{align}

In the next subsections, we investigate the asymptotic behavior of the estimator for an increasing number of mixture components and discuss the complexity and memory requirements of the estimator based on the low-rank structure.

\subsection{Asymptotic Optimality}\label{subsec:optimality}
In \cite[Theorem 2]{9842343}, it is shown that the \ac{cme} approximation via an estimator based on a \ac{gmm} which is learned on the underlying channel distribution is asymptotically converging to the true \ac{cme} for large numbers $K$ of mixture components. 
The key prerequisite for this result is the convergence of the approximate distribution to the true channel distribution as a consequence of the universal approximation property of the \ac{gmm}. 
Due to space limitations we do not restate \cite[Theorem 2]{9842343} in this work.
Building on this result, we can state the asymptotic optimality of the \ac{mfa}-based channel estimator \eqref{eq:mfa_estimator} as a direct consequence of \cite[Theorem 2]{9842343} by using the universal approximation property.
\begin{corollary}\label{cor:optimality}
    Let $p(\B h)$ be any continuous \ac{pdf} which vanishes at infinity. For an arbitrary number of latent dimensions $L$, the \ac{mfa}-based channel estimator \eqref{eq:mfa_estimator} converges to the true \ac{mse}-optimal \ac{cme} \eqref{eq:cme} in the sense that
    \begin{equation}
        \lim_{K\rightarrow \infty}\|\hat{\B h}_{\text{\normalfont CME}}(\B y) - \hat{\B h}^{(K,L)}(\B y)\| = 0
    \end{equation}
    holds for any given $\B y$.
\end{corollary}
\begin{proof}
    The result is a direct consequence of \cite[Theorem 2]{9842343} by using the universal approximation ability \eqref{eq:univ_approx} of the \ac{mfa} model which follows from \cite[Theorem 5]{NgNgChMc20}.
\end{proof}

Corollary \ref{cor:optimality} shows the powerful abilities of the \ac{mfa}-based channel estimator in combination with the possibility to reduce the latent dimension $L$. That said, the practicability of the estimator for a fixed number $L$ of latent dimensions and a finite number $K$ of mixtures is yet to be investigated. We will therefore show simulation results that demonstrate the strong performance of the estimator for real-world data for different numbers of mixtures and latent dimensions in \Cref{sec:sim_results}.

\subsection{Baseline Channel Estimators}\label{subsec:baselines}
We first discuss non-data-based channel estimators, one of whom is the \ac{ls} estimator which simply computes $\hat{\B h}_{\text{LS}} = \B y$ in our case. 
Another technique is compressive sensing, where the channel is assumed to be (approximately) sparse such that we have $\B h \approx \B \Delta \B s$ for a sparse vector \( \mbs \in \C^M \).
The dictionary \( \B \Delta \in \C^{N\times M} \) is typically an oversampled \ac{dft} matrix \cite{AlLeHe15}.
A baseline algorithm is \ac{omp} \cite{PaReKr93} which recovers an estimate \( \hat{\mbs} \) of \( \mbs \) assuming \( \mby = \B \Delta \mbs + \mbn \) and estimates the channel as \( \hat{\B h}_{\text{OMP}} = \B \Delta \hat{\mbs} \) for which \ac{omp} needs to know the sparsity order. Since order estimation is a difficult problem, we use a genie-aided approach: \ac{omp} gets access to the true channel to choose the optimal sparsity. We employ the \ac{omp} algorithm with \( M = 4 N \).


Next, we investigate state-of-the-art data-based algorithms. An important baseline is the \ac{lmmse} formula based on the sample covariance matrix. For this, we use all $T$ training data from $\mathcal{H}$ to compute $\B C = \frac{1}{T}\sum_{t=1}^T \B h_t \B h_t\h$ and then estimate channels as $\hat{\B h}_{\text{LMMSE}} = \B C(\B C + \sigma^2 \eye)\inv \B y$.
A \ac{cnn}-based channel estimator was introduced in~\cite{NeWiUt18} whose architecture is derived via insights about the channel/system model. 
The activation function is the \ac{relu} and we use the \( 2N \times N \) truncated \ac{dft} matrix as input transform, cf. \cite[eq. (43)]{NeWiUt18}. Thereby, an independent \ac{cnn} is trained for each \ac{snr} value.
A related technique to the proposed \ac{mfa} approach is the \ac{gmm}-based channel estimator from \cite{9842343,10051921} which learns a \ac{gmm} of the form $p^{(K)}(\B h) = \sum_{k=1}^K p(k) \mathcal{N}_{\mathbb{C}} (\B h; \B \mu_k, \B C_k)$ via an \ac{em} algorithm to approximate the underlying channel distribution with generally full covariances $\B C_k$. 
In \cite{10051921}, the case of circulant and Toeplitz structured covariances is discussed where the underlying \ac{em} algorithm is adapted to have covariances of the form $\B C_k = \B Q\h \diag(\B c_k) \B Q$
where $\B Q$ is a (truncated) \ac{dft} matrix for the (Toeplitz) circulant case.
The resulting estimator then is of the form 
\begin{align}
    \hat{\B h}_{\text{GMM}} = \sum_{k=1}^K p(k\mid \B y) \left(\B \mu_k + \B C_k \B C_{\B y,k}\inv (\B y - \B \mu_k)\right)
\end{align}
where $\B C_{\B y,k} = \B C_k + \sigma^2\eye$ and $p(k\mid \B y)$ is the responsibility of component $k$ for pilot $\B y$, cf. \cite{9842343,10051921}. In \cite{9940363}, the \ac{gmm}-based estimator was already evaluated on measurement data.

\subsection{Memory and Complexity Analysis}\label{subsec:memory}

\begin{table}[t]
\renewcommand{\arraystretch}{1.5}
\begin{center}
    \begin{tabular}
    { |m{0.14\columnwidth}|m{0.25\columnwidth}|m{0.12\columnwidth}|m{0.12\columnwidth}|m{0.12\columnwidth}|@{}m{0pt}@{} }
     \hline
     \textbf{Name} & \textbf{Parameters} & $L=2$ & $L=8$ & $L=16$
     &\\
     \hline
    MFA &$K(LN + N + 2)$
     & $ 1.24\cdot 10^{4}$  & $ 3.67\cdot 10^{4}$  & $ 6.98\cdot 10^{4}$
     &\\
    \hline
     GMM full  &$K(\frac{1}{2}N^2 + 2N + 1)$
     & \multicolumn{3}{c|}{$1.39\cdot 10^{5}$}
     &\\
    \hline
    GMM toep &$K(5N + 1) $ 
    & \multicolumn{3}{c|}{$2.05\cdot 10^{4}$}
    & \\
    \hline
    GMM circ &$K(2N + 1) $ 
    & \multicolumn{3}{c|}{$8.26\cdot 10^{3}$}
    & \\
    \hline
    \end{tabular}
\end{center}
\caption{Analysis of the number of parameters of the \ac{mfa} model as compared to a (structured) \ac{gmm} with example numbers for $K=N=64$.}
\label{tab:num_params}
\end{table}

In this subsection, we analyze the memory requirements and the computational (online) complexity of the \ac{mfa}-based channel estimator. We emphasize that the fitting of the parameters via the \ac{em} algorithm is done exclusively in an initial offline phase. In the online phase, the channel estimate in \eqref{eq:mfa_estimator} is computed for a given pilot observation $\B y$. We first note that the calculation of the $K$ \ac{lmmse} filters in \eqref{eq:lmmse} as well as of the $K$ responsibilities in \eqref{eq:resp} can be parallelized which is of great importance in real-time systems. 
Furthermore, due to the specific structure of the parameterized covariances, the inverse that appears in \eqref{eq:lmmse} and in \eqref{eq:resp}---for evaluating the Gaussian density---can be computed less expensively by means of the inversion lemma as 
\begin{align}
    (\B W_k \B W_k\h + \B \Psi_k + \sigma^2 \eye)\inv = \B D_k - \B D_k \B W_k \B A_k \B W_k\h \B D_k
\end{align}
where $\B D_k = (\B \Psi_k + \sigma^2 \eye)\inv$ is a diagonal matrix and $\B A_k = (\eye + \B W_k\h \B D_k \B W_k)\inv$ is an $L\times L$ matrix of lower dimension. Thus, the overall order of complexity of the channel estimator is $\mathcal{O}(K(N^2 + NL^2))$ when taking the computation of the matrix products into account. However, since the \ac{lmmse} filters are fixed for a given \ac{snr} value, they can be pre-computed such that only matrix-vector products have to be evaluated. In this case, the complexity reduces to $\mathcal{O}(KN^2)$ which can be computed in $K$ parallel processes.

The memory requirement is determined by the number of parameters for $\left\{\B W_k, \B \Psi_k, \B \mu_k, p(k)\right\}_{k=1}^K$, which depends on the choice of the diagonal covariances $\B \Psi_k$. As discussed also later, the option of having scaled identities $\B \Psi_k = \psi_k^2\eye,~\forall k\in\{1,\dots,K\}$ saves memory overhead and does not affect the channel estimation performance. In this case, the total number of parameters is $K(LN + N +2)$. In Table \ref{tab:num_params}, we compare the number of parameters with the related \ac{gmm} estimator from \cite{9842343,10051921} where we exemplarily depict the number of parameters for a setting with $K=N=64$ for different numbers $L$ of latent dimensions. This particular setting is evaluated in \Cref{sec:sim_results} in terms of the channel estimation performance. 
In comparison to the \ac{gmm} estimator with full covariances, the number of parameters is drastically reduced, which is beneficial against overfitting as shown later. 
Although the \ac{gmm} with structured covariances can have an even lower number of parameters, the performance in this case might suffer from too restrictive assumptions on the underlying structure which is also verified by the simulation results. 
Altogether, the \ac{mfa} model allows for an adaptivity in the number of parameters which reflects a trade-off between memory overhead, computational complexity, and estimation performance which can be conveniently optimized for a certain application scenario.

\section{Simulation Results}\label{sec:sim_results}
We conducted numerical experiments for the data from the measurement campaign as discussed in \Cref{sec:system_model} with $N=64$ antennas. The channels are normalized such that $\op E[\|\B h\|_2^2] = N$, and the SNR is defined as $1/\sigma^2$. The \ac{mse} between the true and estimated channels is normalized by $N$. For the \ac{mfa} model we restrict the covariances to $\B \Psi_k = \psi_k^2\eye~\forall k$. All data-based approaches are trained on the same dataset $\mathcal{H}$.

\begin{figure}[t]
	\centering
	\includegraphics[width=0.95\columnwidth]{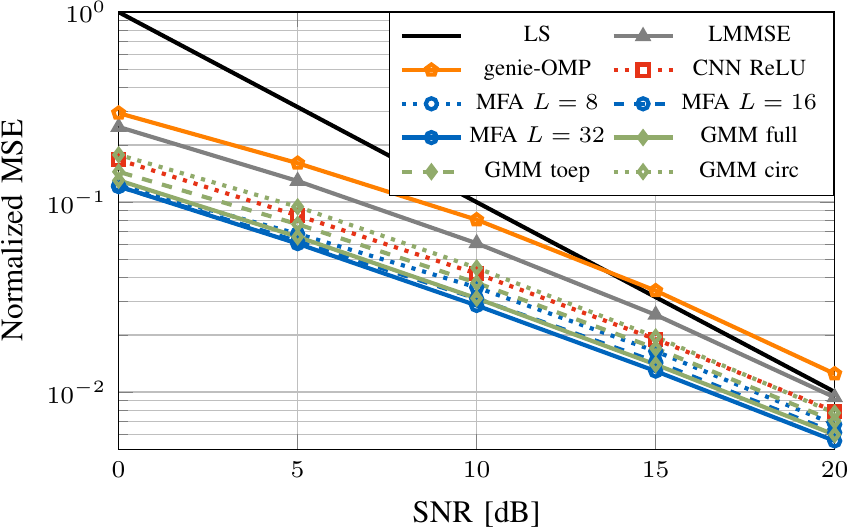}
	\includegraphics[width=0.95\columnwidth]{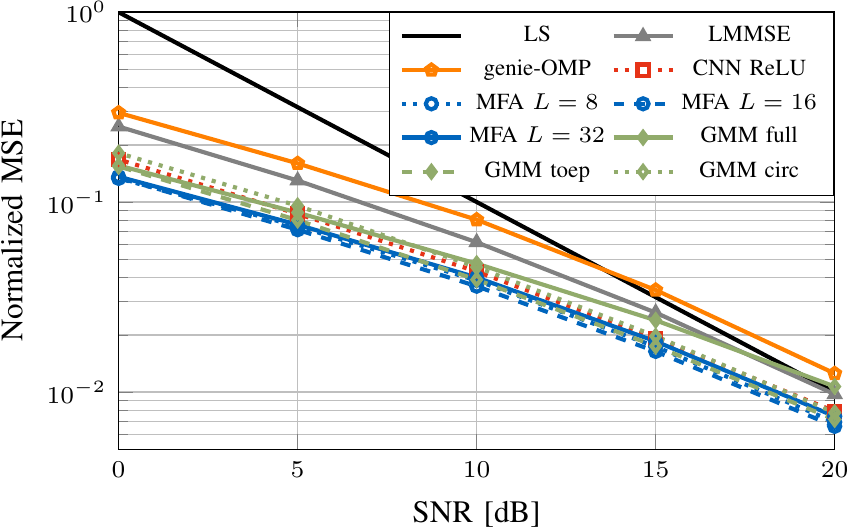}
	\caption{\ac{mse} performance for $T=100{,}000$ (top) and $T=10{,}000$ (bottom) training samples for $K=64$ components.}
	\label{fig:snr}
\end{figure}

Fig. \ref{fig:snr} shows the \ac{mse} performance of the \ac{mfa}-based channel estimator with $K=64$ components in comparison with the baseline estimators introduced in \Cref{subsec:baselines} (the \ac{gmm} variants also have $K=64$ components) for $T=100{,}000$ (top) and $T=10{,}000$ (bottom) training samples. We observe that the approaches ``\ac{ls}'', ``\ac{lmmse}'', and ``genie-\ac{omp}'' as introduced in \Cref{subsec:baselines} do not vary in the performance between both cases. 
The compressive sensing approach genie-\ac{omp}, although having genie knowledge of the sparsity order, does not perform well which indicates that the sparsity assumption in the \ac{dft} dictionary is too restrictive for the measurement data, cf. \cite{9940363}. 
The \ac{cnn}-based channel estimator, labelled ``CNN ReLU'', cf. \Cref{subsec:baselines}, performs better, but is almost always worse than the \ac{gmm}- and \ac{mfa}-based channel estimators, highlighting their strong performance.

For the case of $T=100{,}000$ training samples in Fig. \ref{fig:snr} (top), the \ac{gmm}-based estimator with full covariances (``GMM full'') performs better than the Toeplitz (``GMM toep'') or circulant (``GMM circ'') structured versions, however, it is outperformed by the \ac{mfa}-based estimator with $L=32$ latent dimensions over the whole \ac{snr} range. Reducing the latent dimension further leads to a worse \ac{mse}, but the performance gap is small, especially in the low \ac{snr} regime, and the structured \ac{gmm} variants are still outperformed for all \acp{snr}. 

In Fig. \ref{fig:snr} (bottom), the case of $T=10{,}000$ training samples is shown where the \ac{gmm}-based estimator with full covariances suffers from overfitting, especially in the high \ac{snr} regime, cf. \cite{9842343}. Although the structured \ac{gmm} variants are less prone to overfitting, they are outperformed by the \ac{mfa}-based estimator with $L=8$ or $L=16$ latent dimensions over the whole \ac{snr} range. Increasing the latent dimension to $L=32$ leads to a slightly worse performance due to overfitting.
Altogether, these results highlight the importance of the adaptive nature of the proposed \ac{mfa}-based estimator, which allows to choose a suitable latent dimension for a limited amount of training data.

\begin{figure}[t]
	\centering
	\includegraphics[width=0.95\columnwidth]{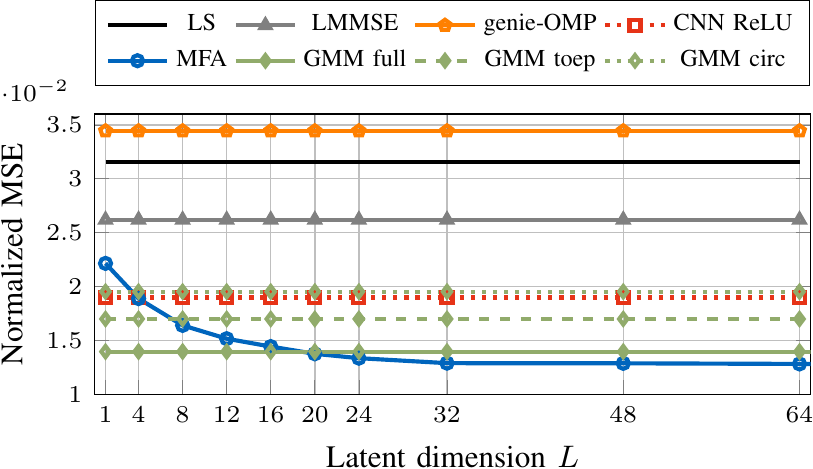}
	
	\vspace{-0.2cm}
	
	\includegraphics[width=0.95\columnwidth]{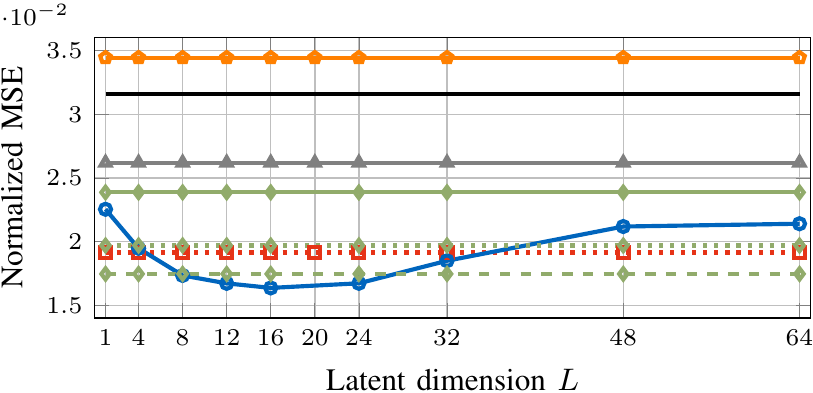}
	
	\vspace{0.2cm}
	
	\includegraphics[width=0.95\columnwidth]{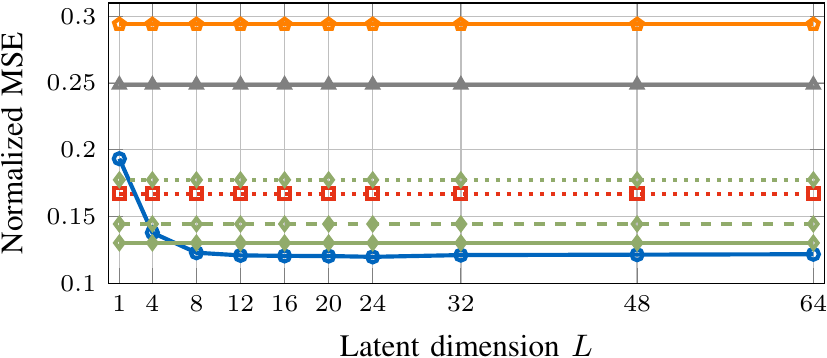}
	\caption{\ac{mse} performance over the number $L$ of latent dimensions for $K=64$ components. Top: $T=100{,}000$ and $\text{SNR} = \SI{15}{dB}$. Middle: $T=10{,}000$ and $\text{SNR} = \SI{15}{dB}$. Bottom: $T=100{,}000$ and $\text{SNR} = \SI{0}{dB}$.}
	\label{fig:latents_15dB}
\end{figure}

In Fig. \ref{fig:latents_15dB}, we further investigate the impact of the latent dimension $L$ on the performance for $K=64$ mixture components and for different numbers of training samples and \ac{snr} values. We compare the results with the baseline estimators which have a constant behavior with respect to the latent dimensions.
In the top plot, the case of $T=100{,}000$ and $\text{SNR} = \SI{15}{dB}$ is evaluated in which the \ac{mse} of the \ac{mfa}-based estimator is consistently decreasing for higher latent dimensions, but with a saturation above $L=32$. The approaches \ac{ls}, \ac{lmmse}, and genie-\ac{omp} are outperformed already with a single latent dimension.
All \ac{gmm}-based variants and the \ac{cnn} estimator are outperformed for $L=20$ latent dimensions or more.
The middle plot shows that for a small amount of training data $T=10{,}000$ and $\text{SNR} = \SI{15}{dB}$, the unstructured \ac{gmm}-based estimator performs even worse than the \ac{mfa} model with only a single latent dimension. Furthermore, increasing the latent dimension up to $L=16$ yields a consistently better \ac{mse} performance, whereas a degradation can be observed for higher latent dimensions due to overfitting effects. In addition, even the structured \ac{gmm}-based variants and the \ac{cnn} estimator are outperformed for $L\in[8,24]$ latent dimensions. 
In the bottom plot, a lower \ac{snr} value of \SI{0}{dB} for $T=100{,}000$ training samples is evaluated. The plot indicates that for lower \ac{snr} values, generally less latent dimensions are necessary since a saturation already occurs above $L=12$. All baseline estimators are outperformed with $L=8$ or more latent dimensions.

\begin{figure}[t]
	\centering
	\includegraphics[width=0.95\columnwidth]{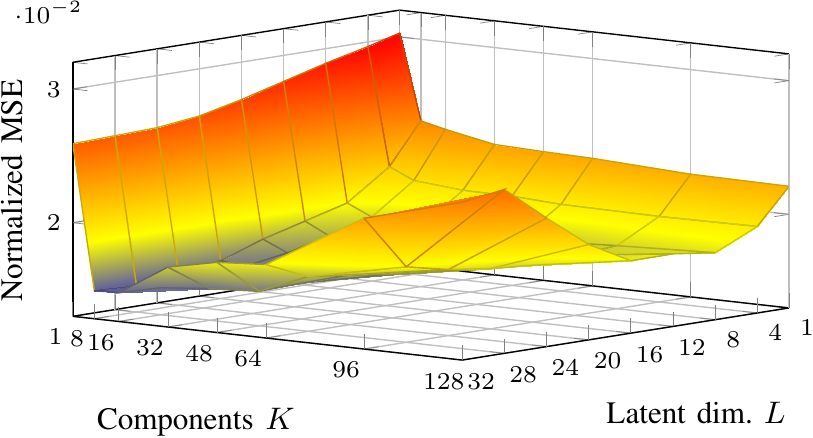}
	\includegraphics[width=0.95\columnwidth]{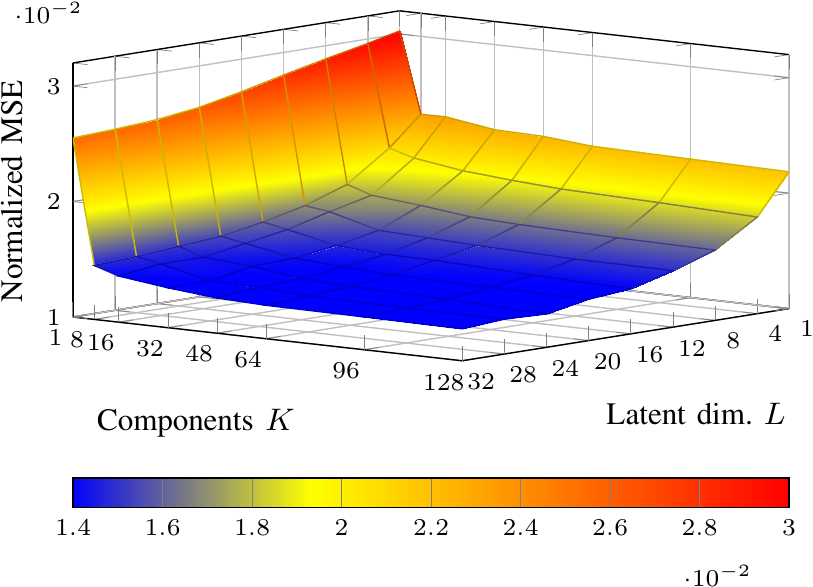}
	\caption{\ac{mse} performance over the number of latent dimensions $L$ and mixture components $K$ for $T=10{,}000$ (top) and $T=100{,}000$ (bottom) training samples for an \ac{snr} of \SI{15}{dB}.}
	\label{fig:latents_comps_3d}
\end{figure}

In Fig. \ref{fig:latents_comps_3d}, the combined impact of both the latent dimensions $L\in[1,32]$ and mixture components $K\in[1,128]$ on the \ac{mse} of the \ac{mfa} model is shown for $\text{SNR} = \SI{15}{dB}$. The case of $K=1$ refers to the \ac{fa} model as described in \Cref{sec:mfa}. In the top plot for $T=10{,}000$ it can be observed that the resulting \ac{mse} is minimal at $(K,L) = (8,28)$ and it increases for higher numbers of latent dimensions and mixture components as expected due to overfitting. In contrast, for $T=100{,}000$ training samples, the \ac{mse} is consistently decreasing for increasing latent dimensions and mixture components with an overall saturation effect. This analysis validates the asymptotic analysis up to a certain extent, since even if the latent dimension is fixed, the \ac{mse} can be further decreased by increasing the number of mixture components $K$. We note that this behavior inherently depends on the available amount of training data.

\section{Conclusion and Outlook}
In this work, we have proposed to employ the \ac{mfa} model for channel estimation which is asymptotically optimal in the sense that the convergence to the \ac{mse}-optimal \ac{cme} is guaranteed. Thereby, the underlying (unknown) channel distribution of a whole communication scenario is learned offline by fitting the \ac{mfa} parameters. 
Afterwards, the \ac{mfa} model is leveraged for online channel estimation where especially the number of parameters and computational complexity is reduced due to the inherent low-rank structure. Simulation results based on measurement data showed great estimation performances.

In future work, we plan to further investigate the \ac{mfa} model for physical layer applications. Particularly interesting is to utilize the subspace information that is provided by the \ac{mfa}, e.g., for interference channels, feedback applications \cite{turan2023versatile}, or in \ac{ris}-aided systems \cite{fesl2023channel}. In addition, we want to adapt the approach to deal with imperfect training data and apply it to channels at higher frequencies.

\bibliographystyle{IEEEtran}
\bibliography{IEEEabrv,biblio}

\begin{thebibliography}{10}
\providecommand{\url}[1]{#1}
\csname url@samestyle\endcsname
\providecommand{\newblock}{\relax}
\providecommand{\bibinfo}[2]{#2}
\providecommand{\BIBentrySTDinterwordspacing}{\spaceskip=0pt\relax}
\providecommand{\BIBentryALTinterwordstretchfactor}{4}
\providecommand{\BIBentryALTinterwordspacing}{\spaceskip=\fontdimen2\font plus
\BIBentryALTinterwordstretchfactor\fontdimen3\font minus
  \fontdimen4\font\relax}
\providecommand{\BIBforeignlanguage}[2]{{%
\expandafter\ifx\csname l@#1\endcsname\relax
\typeout{** WARNING: IEEEtran.bst: No hyphenation pattern has been}%
\typeout{** loaded for the language `#1'. Using the pattern for}%
\typeout{** the default language instead.}%
\else
\language=\csname l@#1\endcsname
\fi
#2}}
\providecommand{\BIBdecl}{\relax}
\BIBdecl

\bibitem{Roweis1999}
S.~Roweis and Z.~Ghahramani, ``A {Unifying} {Review} of {Linear} {Gaussian}
  {Models},'' \emph{Neural Comput.}, vol.~11, no.~2, pp. 305--345, Feb. 1999.

\bibitem{NgNgChMc20}
T.~T. Nguyen, H.~D. Nguyen, F.~Chamroukhi, and G.~J. McLachlan,
  ``{Approximation by Finite Mixtures of Continuous Density Functions that
  Vanish at Infinity},'' \emph{Cogent Math. Statist.}, vol.~7, no.~1, p.
  1750861, 2020.

\bibitem{Kingma2014}
D.~P. Kingma and M.~Welling, ``{Auto-Encoding Variational Bayes},'' in
  \emph{Proc. 2nd Int. Conf. Learn. Represent.}, 2014.

\bibitem{10056957}
N.~Shlezinger, J.~Whang, Y.~C. Eldar, and A.~G. Dimakis, ``{Model-Based Deep
  Learning},'' \emph{Proc. IEEE}, pp. 1--35, 2023.

\bibitem{9598915}
M.~Alsabah \emph{et~al.}, ``{6G Wireless Communications Networks: A
  Comprehensive Survey},'' \emph{IEEE Access}, vol.~9, pp. 148\,191--148\,243,
  2021.

\bibitem{9842343}
M.~Koller, B.~Fesl, N.~Turan, and W.~Utschick, ``{An Asymptotically MSE-Optimal
  Estimator Based on Gaussian Mixture Models},'' \emph{IEEE Trans. Signal
  Process.}, vol.~70, pp. 4109--4123, 2022.

\bibitem{10051858}
M.~Baur, B.~Fesl, M.~Koller, and W.~Utschick, ``{Variational Autoencoder
  Leveraged MMSE Channel Estimation},'' in \emph{56th Asilomar Conf. Signals,
  Syst., Comput.}, 2022, pp. 527--532.

\bibitem{10051921}
B.~Fesl, M.~Joham, S.~Hu, M.~Koller, N.~Turan, and W.~Utschick, ``{Channel
  Estimation based on Gaussian Mixture Models with Structured Covariances},''
  in \emph{56th Asilomar Conf. Signals, Syst., Comput.}, 2022, pp. 533--537.

\bibitem{10078293}
B.~Fesl, N.~Turan, M.~Joham, and W.~Utschick, ``{Learning a Gaussian Mixture
  Model from Imperfect Training Data for Robust Channel Estimation},''
  \emph{IEEE Wireless Commun. Lett.}, 2023.

\bibitem{Murphy2012}
K.~P. Murphy, \emph{Machine {Learning}: {A} {Probabilistic}
  {Perspective}}.\hskip 1em plus 0.5em minus 0.4em\relax The MIT Press, 2012.

\bibitem{em_mfa}
Z.~Ghahramani and G.~E. Hinton, ``{The EM Algorithm for Mixtures of Factor
  Analyzers},'' University of Toronto, Tech. Rep., 1996.

\bibitem{deepMFA}
Y.~Tang, R.~Salakhutdinov, and G.~Hinton, ``{Deep Mixtures of Factor
  Analysers},'' in \emph{Proc. 29th Int. Conf. Mach. Learn.}, 2012, p.
  1123–1130.

\bibitem{8645457}
D.~Ramírez, I.~Santamaria, S.~Van~Vaerenbergh, and L.~L. Scharf, ``{An
  Alternating Optimization Algorithm for Two-Channel Factor Analysis with
  Common and Uncommon Factors},'' in \emph{52nd Asilomar Conf. Signals, Syst.,
  Comput.}, 2018, pp. 1743--1747.

\bibitem{7727995}
H.~Xie, F.~Gao, and S.~Jin, ``{An Overview of Low-Rank Channel Estimation for
  Massive MIMO Systems},'' \emph{IEEE Access}, vol.~4, pp. 7313--7321, 2016.

\bibitem{AlLeHe15}
A.~Alkhateeb, G.~Leus, and R.~W. Heath, ``{Compressed Sensing based Multi-User
  Millimeter Wave Systems: {How} Many Measurements are needed?}'' in \emph{IEEE
  Int. Conf. Acoust., Speech, Signal Process. (ICASSP)}, 2015, pp. 2909--2913.

\bibitem{PaReKr93}
Y.~C. Pati, R.~Rezaiifar, and P.~S. Krishnaprasad, ``{Orthogonal Matching
  Pursuit: {Recursive} Function Approximation with Applications to Wavelet
  Decomposition},'' in \emph{Proc. 27th Asilomar Conf. Signals, Syst.,
  Comput.}, Nov. 1993, pp. 40--44 vol.1, iSSN: 1058-6393.

\bibitem{NeWiUt18}
D.~Neumann, T.~Wiese, and W.~Utschick, ``{Learning the {MMSE} Channel
  Estimator},'' \emph{{IEEE} Trans. Signal Process.}, vol.~66, no.~11, pp.
  2905--2917, Jun. 2018.

\bibitem{9940363}
N.~Turan, B.~Fesl, M.~Grundei, M.~Koller, and W.~Utschick, ``{Evaluation of a
  Gaussian Mixture Model-based Channel Estimator using Measurement Data},'' in
  \emph{Int. Symp. Wireless Commun. Syst. (ISWCS)}, 2022.

\bibitem{turan2023versatile}
N.~Turan, B.~Fesl, M.~Koller, M.~Joham, and W.~Utschick, ``{A Versatile
  Low-Complexity Feedback Scheme for FDD Systems via Generative Modeling},''
  2023, arXiv preprint: 2304.14373.

\bibitem{fesl2023channel}
B.~Fesl, A.~Faika, N.~Turan, M.~Joham, and W.~Utschick, ``{Channel Estimation
  with Reduced Phase Allocations in RIS-Aided Systems},'' 2023, arXiv preprint:
  2211.07552.

\end{thebibliography}

\end{document}